\documentclass[reqno]{amsart}%
\usepackage{amsfonts}
\usepackage{verbatim}
\usepackage{xcolor}
\usepackage{amsmath}
\usepackage{amssymb}
\usepackage{graphicx}
\usepackage{accents}%
\providecommand{\U}[1]{\protect\rule{.1in}{.1in}}
\newtheorem{theorem}{Theorem}[section]

\newtheorem{corollary}[theorem]{Corollary}
\newtheorem{definition}[theorem]{Definition}

\newtheorem{hypothesis}[theorem]{Hypothesis}
\numberwithin{equation}{section}
\begin{document}
\title[KdV equation]{The inverse scattering transform for weak Wigner-von Neumann type potentials}
\author{Sergei Grudsky}
\address{Departamento de Matematicas, CINVESTAV del I.P.N. Aportado Postal 14-740,
07000 Mexico, D.F., Mexico.}
\email{grudsky@math.cinvestav.mx.}
\author{Alexei Rybkin}
\address{Department of Mathematics and Statistics, University of Alaska Fairbanks, PO
Box 756660, Fairbanks, AK 99775}
\email{arybkin@alaska.edu}
\thanks{The first author is supported by CONACYT grant 61517/2020.}
\thanks{The second author is supported in part by the NSF under grant DMS 1716975.}
\date{February 2022}

\begin{abstract}
In the context of the Cauchy problem for the Korteweg-de Vries equation we
extend the inverse scattering transform to initial data that behave at plus
infinity like a sum of Wigner-von Neumann type potentials with small coupling
constants. Our arguments are based on the theory of Hankel operators.

\end{abstract}
\dedicatory{In memory of Serguei Naboko.}\maketitle

\section{Introduction}

The present paper is concerned with extension of the inverse scattering
transform (IST) for the Cauchy problem for the Korteweg-de Vries (KdV)
equation
\begin{equation}
\partial_{t}q-6q\partial_{x}q+\partial_{x}^{3}q=0,\label{KdV}%
\end{equation}%
\begin{equation}
q\left(  x,0\right)  =q\left(  x\right)  ,\label{KdVID}%
\end{equation}
with long-range initial profiles $q$. In our \cite{GruRybSIMA15} we have
disposed of essentially any decay condition\footnote{Which means that
(\ref{KdV}) does not require a boundary condition at $-\infty$ .} on $q$ at
$-\infty$ but still required short-range decay at $+\infty$. This should not
come as a surprise since the KdV is a strongly unidirectional equation
(solitons run to the right and radiation waves run to the left) which should
translate into different contributions from the behavior of the data $q$ at
$\pm\infty$ . In the present paper we put forward a general framework which is
robust enough to deal with long range decay at $+\infty$. More specifically,
we consider decay that is slow enough to include physically relevant
Wigner-von Neumann (WvN) resonances (also referred to as spectral
singularities) but those resonances are in a way weak (spectral singularities
of low order). We note that a WvN resonance could be an embedded eigenvalue as
was first shown by Wigner-von Neumann in their seminal paper \cite{WvN1929}.
We refer the interested reader to recent \cite{Kruger12,Lukic13} and the
extensive literature on embedded spectra cited therein. Note that under our
assumptions (see Hypothesis \ref{hyp1.1}) embedded eigenvalues do not appear
and the conditions of Hypothesis \ref{hyp1.1} are dictated by what we know
about WvN resonances (and long-range potentials as a whole). Our conditions
can be relaxed as we fill in the gaps in our understanding of long-range
scattering (discussed further in section \ref{apps}). Of course the spectrum
in the long-range situation can be so rich and complicated that it is too long
of a shot to understand it to our complete satisfaction.

To fix our notation we give a brief review of the classical IST. Recall (see,
e.g. \cite{AC91, March86,NovikovetalBook}) that conceptually the IST is
similar to the Fourier transform and consists of\ three steps:

\textbf{Step 1. }(direct transform)
\[
q\left(  x\right)  \longrightarrow S_{q}\text{,}%
\]
where $S_{q}$ is a new set of variables which turns (\ref{KdV}) into a simple
first order linear ODE for $S_{q}(t)$ with the initial condition
$S_{q}(0)=S_{q}$.

\textbf{Step 2. }(time evolution)
\[
S_{q} \longrightarrow S_{q} \left( t\right) \text{.}%
\]

\textbf{Step 3. }(inverse transform)
\[
S_{q} \left( t\right)  \longrightarrow q (x ,t)\text{.}%
\]

Similar methods have also been developed for many other evolution nonlinear
PDEs, which are referred to as completely integrable \ Each of steps 1-3
involves solving a linear integral equation that allows us to analyze
integrable systems at the level unreachable by neither direct numerical
methods nor standard PDE techniques.

This formalism works smoothly in two basic cases (we refer to as classical):
when
\begin{equation}
\int_{\mathbb{R}}\left(  1+\left\vert x\right\vert \right)  q\left(  x\right)
\mathrm{d}x<\infty\text{ (the short-range case)}\label{sr}%
\end{equation}
and when $q$ is periodic. We will only be concerned with the
former\footnote{For the latter we refer to \cite{GesHold03}.} which goes as
follows. Associate with $q$ the full line (self-adjoint) Schrodinger operator
$\mathbb{L}_{q}=-\partial_{x}^{2}+q(x)$. For its spectrum $\sigma\left(
\mathbb{L}_{q}\right)  $ we have%
\[
\sigma\left(  \mathbb{L}_{q}\right)  =\sigma_{d}\left(  \mathbb{L}_{q}\right)
\cup\sigma_{ac}\left(  \mathbb{L}_{q}\right)  \text{,}%
\]
where the discrete component $\sigma_{d}\left(  \mathbb{L}_{q}\right)
=\{-\kappa_{n}^{2}\}$ is finite and for the absolutely continuous one has
$\sigma_{ac}\left(  \mathbb{L}_{q}\right)  =[0,\infty)$. There is no singular
continuous spectrum. The Schr{\"{o}}dinger equation
\begin{equation}
\mathbb{L}_{q}\psi=k^{2}\psi\label{SE}%
\end{equation}
has two (linearly independent) Jost solutions $\psi_{\pm}(x,k)$, i.e.
solutions satisfying%
\begin{equation}
\psi_{\pm}(x,k)=\mathrm{e}^{\pm\mathrm{i}kx}+o(1),\;\partial_{x}\psi_{\pm
}(x,k)\mp\mathrm{i}k\psi_{\pm}(x,k)=o(1),\text{\ \ }x\rightarrow\pm
\infty.\label{Jost solutions}%
\end{equation}
Since $q$ is real, $\overline{\psi}_{+}$ also solves (\ref{SE}) and one can
easily see that the pair $\{\psi_{+},\overline{\psi}_{+}\}$ forms a
fundamental set for (\ref{SE}). Hence $\psi_{-}$ is a linear combination of
$\{\psi_{+},\overline{\psi}_{+}\}$. We write this fact as follows
($k\in\mathbb{R}$)%
\begin{equation}
T(k)\psi_{-}(x,k)=\overline{\psi_{+}(x,k)}+R(k)\psi_{+}(x,k),\text{ (basic
scattering identity)}\label{R basic scatt identity}%
\end{equation}
where $T$ and $R$ are called the\emph{ transmission} and (right)\emph{
reflection} coefficient respectively. The identity
(\ref{R basic scatt identity}) is totally elementary but serves as a basis for
inverse scattering theory. As is well-known (see, e.g. \cite{March86}), the
triple%
\begin{equation}
S_{q}=\{R,(\kappa_{n},c_{n})\},\label{SD}%
\end{equation}
where $c_{n}=\left\Vert \psi_{+}(\cdot,\mathrm{i}\kappa_{n})\right\Vert ^{-1}%
$, determines $q$ uniquely and is called the scattering data for
$\mathbb{L}_{q}$ . Computing $S_{q}$ ends step 1. Note that this set can be
constructed for any $q\in L^{1}$ with possibly infinitely many $(\kappa
_{n},c_{n})$.

Step 2 is computationally easy. Lax pair considerations readily yield
\begin{equation}
S_{q}(t)=\left\{  R(k)\exp\left(  8\mathrm{i}k^{3}t\right)  ,\;\kappa
_{n},c_{n}\exp\left(  8\kappa_{n}^{3}t\right)  \right\}  .\label{time evol}%
\end{equation}
Step 3 amounts to solving the inverse scattering problem of recovering the
potential $q\left(  x,t\right)  $ (which now depends on $t\geq0$ ) from
$S_{q}(t)$ via any of the three main methods: the Gelfand-Levitan-Marchenko
equation, the trace formula, and the Riemann-Hilbert problem (put in the
historic order).

We emphasize that while steps 1-2 remain valid for essentially any real $q$
supporting Jost solutions $\psi_{\pm}\left(  x,k\right)  $, step 3 breaks down
in general for long-range potentials (i.e potentials for which (\ref{sr}) does
not hold). In fact, the latter occurs already in the case of $q\left(
x\right)  =O\left(  x^{-2}\right)  ,x\rightarrow\pm\infty$ (i.e. slightly
worse than (\ref{sr})). The first example to this effect is explicitly
constructed in \cite{ADM81}. The authors present two distinct potentials with
such decay and no bound states which share the same reflection coefficient.
Thus we have non-uniqueness in solving step 3. The reason for the loss of
uniqueness is not explained in \cite{ADM81} but it seems plausible that an
approximation of the Jost solution by a suitable Bessel function would capture
the singular behavior of $\psi_{\pm}\left(  x,k\right)  $ near $k=0$ (the edge
of the a.c. spectrum) which in turn might suggest what pieces of data need to
be added to $S_{q}$ to resolve the issue\footnote{We emphasize that $x^{-2}$
decay may produce infinite (negative) discrete spectrum.} (see
\cite{AktKlaus93} for some relevant results). For potentials $q\left(
x\right)  =O\left(  x^{-\alpha}\right)  ,1<\alpha<2$, as is shown in
\cite{Yafaev81}, even nice potentials have the Jost solution with an erratic
behavior at $k=0$. In some terminology, real points of discontinuity of the
Jost solution are referred to as \emph{spectral singularities}. Since the
absence of such points is one of the main assumptions in the inverse
scattering method, it is reasonable to link non-uniqueness in the case of
$q\left(  x\right)  =O\left(  x^{-\alpha}\right)  ,1<\alpha\leq2$, to a
spectral singularity at $k=0$. $\text{{}}$Note that even though $k=0$ is the
only possible spectral singularity for a summable potential it is a good open
problem \cite{AK01} to find the extra data that restore uniqueness. If
$q\notin L^{1}$ then the situation becomes even worse as (\ref{SE}) need not
have a Jost solution and all steps 1-3 fail in general. Extending the inverse
scattering procedure to such potentials is currently out of reach and any
essential step towards its solution is important.

Our aim here is to identify a class of $L^{2}$ non-integrable potentials
supporting a rich set of spectral singularities which nevertheless can be
uniquely restored by the classical scattering data (\ref{SD}). We came across
this class while studying its well-known representatives, oscillatory
potentials of the form
\begin{equation}
q_{\gamma}\left(  x\right)  =(A/x)\sin2\omega x,\text{\ \ \ }\gamma
:=\left\vert A/\left(  4\omega\right)  \right\vert ,\label{pure WvN}%
\end{equation}
typically referred to as WvN type potentials. The scattering theory can be
developed along the same lines with its short-range counterpart except for the
fact the points $k=\pm\omega$ are spectral singularities of order $\gamma$
\cite{Klaus91}. In the WvN potential community the point $\omega^{2}$ is
called a \emph{WvN resonance}. In our recent \cite{RyNON21} we explicitly
construct three potentials with the same (single) resonance $\omega^{2}$ and
the same set (\ref{SD}). Note that $k=0$ is not a spectral singularity in
those examples meaning that WvN resonances may also cause non-uniqueness. The
main goal of our note is to show that for small $\gamma$ (weak coupling
constant) non-uniqueness may only come from a spectral singularity at zero.

Our interest in WvN potentials is inspired in part by the work of Matveev (see
\cite{Matveev2002} and the literature cited therein) and his proposal
\cite{MatveevOpenProblems}: "A very interesting unsolved problem is to study
the large time behavior of the solutions to the KdV equation corresponding to
the smooth initial data like $cx^{-1}\sin2kx$, $c\in\mathbb{R}$. Depending on
the choice of the constant $c$ the related Schr{\"{o}}dinger operator might
have finite or infinite or zero number of the negative eigenvalues. The
related inverse scattering problem is not yet solved and the study of the
related large times evolution is a very challenging problem."

In the present paper we deal with initial data (\ref{KdVID}) subject to

\begin{hypothesis}
\label{hyp1.1} Let $q$ be a real locally integrable function subject to

\begin{enumerate}
\item[(1)] (decay) $\mathbb{L}_{q}$ supports Jost solutions $\psi_{\pm}$ such
that\footnote{Here $H^{2}$ stands for the Hardy space of the upper half plane
(see Section \ref{hankel})}%
\[
Y_{\pm}(x,k):=\mathrm{e}^{\mp\mathrm{i}kx}\psi_{\pm}(x,k+\mathrm{i}0)-1\in
H^{2};
\]

\end{enumerate}

\begin{enumerate}
\item[(2)] (piecewise continuity of the reflection coefficient) $R$ is
piecewise continuous on $\mathbb{R}\setminus\left\{  \pm\omega_{j}\right\}  $
and has a jump discontinuity at each $\pm\omega_{j}$;

\item[(3)] (jump size)
\[
\sup_{j}\frac{1}{2} \left\vert R \left(  \pm\omega_{j} +0\right)  -R \left(
\pm\omega_{j} -0\right) \right\vert <1 ;
\]

\item[(4)] ( $L^{2}$ type decay) $R \in L^{2} ;$

\item[(5)] (discrete spectrum) $\sigma_{d} \left( \mathbb{L}_{q}\right) $ is finite.
\end{enumerate}
\end{hypothesis}

The set of such potentials is quite broad but its description in terms of $q$
alone is out of reach. It can however be describes (or perhaps even
characterized) via the spectral measure of the underlying Schrodinger operator
$\mathbb{L}_{q}$. Such descriptions though are rarely useful in the IST
context as steps 1-3 stated in terms of spectral measure work well only in the
setting of periodic potentials and some of their generalizations (see the
recent \cite{Binderetal2018} for the spectral approach to almost periodic
potentials). For this reason we only outline here how it can be done. The
details will be provided elsewhere in a more general case.

Recall that the spectral measure (or rather matrix) of $\mathbb{L}_{q}$ can be
obtained in terms of the (scalar) spectral measures $\rho_{\pm}$ associated
with two half-line Schrodinger operators with (say) a Dirichlet boundary
condition at (say) zero (see e.g. \cite{Titchmarsh62}). In the context of
$L^{2}$ potentials (our potentials are $L^{2}$) such measures are completely
characterized in the beautiful paper \cite{KillipSimon2008}. It is important
that three out of four conditions in this characterization follow from the
Zakharov-Faddeev trace formula (\ref{trace formula}) below, which is extended
in \cite{KillipSimon2008} to $L^{2}$ potentials, all conditions being very
mild. On the other hand, detailed descriptions of the the Jost solution, the
Weyl m-function, and the spectral measure\footnote{These functions are of
course closely related.} of the half-line (Dirichlet) Schrodinger operator
with WvN type potentials (and their sums, finite or even infinite) are given
in \cite{Hintonetal1991, Klaus91} (see section \ref{apps}). This way we
construct two spectral measures $\rho_{\pm}\left(  E\right)  $ with a
desirable behavior for $E>\varepsilon>0$. It remains to prescribe a "correct"
behavior of $\rho_{\pm}\left(  E\right)  $ for $E<\varepsilon$. This is not
done in \cite{Hintonetal1991, Klaus91} but it is where the Killip-Simon
characterization \cite{KillipSimon2008} comes in handy as it allows us to
assign only finitely many pure negative points (eigenvalues) and a desirable
behavior at $E=0$, all within $L^{2}$ potentials (and without altering
$\rho_{\pm}\left(  E\right)  $ on $\left(  \varepsilon,\infty\right)  $).
Consequently, we now have two half-line Weyl m-functions $m_{\pm}$ and hence
the ($2\times2$) Weyl matrix $M$ which representing measure is the spectral
measure of $\mathbb{L}_{q}$ with all conditions of Hypothesis \ref{hyp1.1}
satisfied. Then one can find $q$ on each $\mathbb{R}_{\pm}$ from $\rho_{\pm}$
via the inverse spectral method of Gelfand-Levitan-Marchenko \cite{March86} or
the whole $q$ via its full line adaptation \cite{AK01}. As was mentioned
above, this however does not yield a convenient description of the conditions
in Hypothesis \ref{hyp1.1} in terms of $q$ itself.

\begin{theorem}
[Main Theorem]\label{MainThm} Under Hypothesis \ref{hyp1.1} the data $S_{q}$
given by (\ref{SD}) determine $q$ uniquely.
\end{theorem}

Thus, the principle value of Theorem \ref{MainThm} is that it takes into
account the effect of nonzero resonances (spectral singularities) in the
inverse scattering problem. This theorem can be applied to such physically
interesting cases as certain potentials supporting finitely many WvN
resonances with $\gamma<1/2$ or some potentials of the form $q \left(
x\right)  =p \left( x\right) /x$ where $p \left( x\right) $ is a periodic
function with a zero mean. However, as we discussed above any explicit
description of our class (i.e. in terms of $q$) is out of reach unless we have
a suitable description of $\psi_{ \pm} \left( x ,k\right) $ as $k
\rightarrow0$.

We follow standard notation accepted in Analysis: $\mathbb{R}$ is the real
line, $\mathbb{R}_{\pm}=(0,\pm\infty)$, $\mathbb{C}$ is the complex plane,
$\mathbb{C}^{\pm}=\left\{  z\in\mathbb{C}:\pm\operatorname*{Im}z>0\right\}  $.
$\overline{z}$ is the complex conjugate of $z$. Besides number sets, black
board bold letters will also be used for (linear) operators. In particular,
$\mathbb{I}$ denotes the identity operator. We write $f\left(  x\right)  \sim
g\left(  x\right)  ,x\rightarrow x_{0}\text{,}$ ( $x_{0}$ may be infinite) if
$\lim\left(  f\left(  x\right)  -g\left(  x\right)  \right)  =0,x\rightarrow
x_{0}$.

The paper is organized as follows. In Section \ref{hankel} we give a brief
introduction to the theory of Hankel operators and state explicitly what will
be used. In Section \ref{proof} we give a detailed proof of the main theorem.
Section \ref{apps} is devoted to applications of the main theorem to WvN type
potentials. In the final Section \ref{RHP section} we apply the main theorem
to the analytic factorization problem of the Riemann-Hilbert problem arising
in the IST for the KdV equation.

\section{Hankel Operators, basic definitions and important facts
\label{hankel}}

Our approach is based upon techniques of the Hankel operator. Since the Hankel
operator is not a conventional tool in inverse problems for the reader's
convenience we give a brief introduction to the theory of Hankel operators and
statements of what will be used.

Recall (see e.g. \cite{Garnett}) that a function $f$ analytic in $\mathbb{C}^{
\pm}$ is in the Hardy space $H^{2} \left( \mathbb{C}^{ \pm}\right) $ if
\[
\Vert f\Vert_{2}^{2}:=\sup_{y >0}\int_{\mathbb{R}}\left\vert f (x
\pm\mathrm{i} y)\right\vert ^{2} \mathrm{d} x <\infty\text{.}%
\]
We set $H^{2} =H^{2} \left( \mathbb{C}^{ +}\right) $ . It is a fundamental
fact that $f \left( z\right)  \in H^{2} \left( \mathbb{C}^{ \pm}\right) $ has
non-tangential boundary values $f \left( x \pm\mathrm{i} 0\right) $ for almost
every (a.e.) $x \in\mathbb{R}$ and $H^{2} \left( \mathbb{C}^{ \pm}\right) $
are Hilbert spaces with the inner product induced from $L^{2}$ :
\[
\langle f ,g \rangle_{H_{ \pm}^{2}} = \langle f ,g \rangle=\int_{\mathbb{R}}f
\left( x\right)  \bar{g} \left( x\right)  \mathrm{d} x\text{.}%
\]
It is well-known that $L^{2} =H^{2} \left( \mathbb{C}^{ +}\right)  \oplus
H^{2} \left( \mathbb{C}^{ -}\right) \text{,}$ the orthogonal (Riesz)
projection $\mathbb{P}_{ \pm}$ onto $H^{2} \left( \mathbb{C}^{ \pm}\right) $
being given by
\begin{equation}
(\mathbb{P}_{ \pm} f) (x) = \pm\frac{1}{2 \pi\mathrm{i}} \int_{\mathbb{R}%
}\frac{f (s) \mathrm{d} s}{s -(x \pm\mathrm{i} 0)} .\label{eq3.1}%
\end{equation}

A Hankel operator is an infinitely dimensional analog of a Hankel matrix, a
matrix whose $(j ,k)$ entry depends only on $j +k$ . In the context of
integral operators the Hankel operator is usually defined as an integral
operator on $L^{2} (\mathbb{R}_{ +})$ whose kernel depends on the sum of the
arguments
\begin{equation}
(\mathbb{H} f) (x) =\int_{0}^{\infty}h (x +y) f (y) \mathrm{d} y ,\;f \in
L^{2} (\mathbb{R}_{ +}) ,\;x \geq0 ,\label{eq4.10}%
\end{equation}
and it is this form that Hankel operators typically appear in the inverse
scattering formalism. We however consider Hankel operators on $H^{2}$ (c.f.
\cite{Nik2002,Peller2003}).

Let $(\mathbb{J} f) (x) =f ( -x)$ be the reflection operator in $L^{2}$ . It
is clearly an isometry with the obvious property
\begin{equation}
\mathbb{J} \mathbb{P}_{ \mp} =\mathbb{P}_{ \pm} \mathbb{J}\label{eq4.9}%
\end{equation}

\begin{definition}
[Hankel operator]\label{def4.1} Let $\varphi\in L^{\infty}$ . The operator
$\mathbb{H} (\varphi)$ defined by
\begin{equation}
\mathbb{H} (\varphi) f =\mathbb{J} \mathbb{P}_{ -} (\varphi f) ,\;f \in H^{2}
,\label{eq4.1}%
\end{equation}
is called the Hankel operator with the symbol $\varphi$ .
\end{definition}

It immediately follows from the definition and (\ref{eq4.9}) that
$\Vert\mathbb{H} (\varphi)\Vert\leq\Vert\varphi\Vert_{\infty}$ and if
$\mathbb{J} \varphi=\overline{\varphi}$ then $\mathbb{H} (\varphi)$ is
selfadjoint on $H^{2}$ . A more subtle statement (Hartman's theorem) says that
$\mathbb{H} (\varphi)$ is compact iff $\varphi=\varphi_{1} +\varphi_{2}$ where
$\varphi_{1}$ is a function continuous on the closed real line and
$\varphi_{2}$ is analytic and uniformly bounded in the upper half plane.
However if $\varphi$ is piecewise continuous with jump discontinuities then
$\mathbb{H} (\varphi)$ is no longer compact. The following deep theorem
\cite{Power1982} plays a crucial role in our consideration.

\begin{theorem}
[Power, 1978]\label{Power thm} Let $\mathbb{J}\varphi=\bar{\varphi}$ (i.e.
$\mathbb{H}(\varphi)$ is selfadjoint), $\varphi$ decay at $\pm\infty$. Suppose
that $\varphi$ is piecewise continuous away from some points $\left\{
\pm\omega_{j}\right\}  $ (may be infinitely many) and
\[
\alpha_{j}:=\frac{1}{2}\left\vert \varphi\left(  \omega_{j}+0\right)
-\varphi\left(  \omega_{j}-0\right)  \right\vert
\]
(assuming that the limits exist and are different). Then for the essential
spectrum of $\mathbb{H}(\varphi)$ we have
\[
\sigma_{ess}\left(  \mathbb{H}(\varphi)\right)  =\cup_{j}\left[  -\alpha
_{j},\alpha_{j}\right]  \text{.}%
\]

\end{theorem}

The actual statement of Theorem \ref{Power thm} given in \cite{Power1982} is
more general.

Note that the Hankel operator $\mathbb{H}$ defined by (\ref{eq4.10}) is
unitary equivalent to $\mathbb{H} (\varphi)$ with the symbol $\varphi$ equal
to the Fourier transform of $h$ . We emphasize though that the form
(\ref{eq4.10}) does not prove to be convenient for our purposes and also $h$
is in general not a function but a distribution. In integral form
(\ref{eq4.10}) Hankel operators appeared naturally already in the classical
papers of Faddeev and Marchenko \cite{March86} on inverse scattering. However,
the well-developed theory of this class of operators (and even the name
itself) was not used at that time. In the KdV context, the Fredholm
determinant of $\mathbb{I} +\mathbb{H}$ appears to be studied for the first
time in \cite{Popper84}. In recent \cite{Blower21,Malham21} some ideas of
\cite{Popper84} were extended far beyond the KdV case.

In the conclusion of this section we mention that the formulation of
Faddeev--Marchenko inverse scattering theory in terms of Hankel operators
defined as in Definition \ref{def4.1} (adjusted to the unit circle) and the
techniques stemming from this theory appeared only in the present century in
\cite{Yuditetal2011,Volbergetal2002}. However, the inverse scattering problem
was studied therein only for Jacobi operators and not in the context of
integrable systems.

\section{Proof of the Main Theorem \label{proof}}

Rewrite (\ref{R basic scatt identity}) in the form%
\begin{equation}
Ty_{-}=\bar{y}_{+}+R_{x}y_{+},\label{eq6.15}%
\end{equation}
where
\begin{equation}
y_{\pm}(x,k):=\mathrm{e}^{\mp\mathrm{i}kx}\psi_{\pm}(x,k),\text{\ \ \ }%
R_{x}\left(  k\right)  :=\mathrm{e}^{2\mathrm{i}kx}R\left(  k\right)
.\label{y}%
\end{equation}
The function $y:=y_{+}$ will be used more frequently. Let us regard
(\ref{eq6.15}) as a Hilbert-Riemann problem of determining $y_{\pm}$ by given
$T,R$ which we solve by Hankel operator techniques. The potential $q$ can then
be easily found from either $y_{\pm}$.

As is well-known, for real $k$ we have
\begin{equation}
T(-k)=\overline{T(k)},\text{\ \ \ }R(-k)=\overline{R(k)},\text{\ \ \ }%
\left\vert T\left(  k\right)  \right\vert ^{2}+\left\vert R\left(  k\right)
\right\vert ^{2}=1,\label{eq6.14}%
\end{equation}
and
\[
T\left(  k\right)  =\prod_{n}\frac{k+\mathrm{i}\kappa_{n}}{k-\mathrm{i}%
\kappa_{n}}\exp\left[  \frac{1}{2\pi\mathrm{i}}\int_{\mathbb{R}}\frac
{\log\left(  1-\left\vert R\left(  s\right)  \right\vert ^{2}\right)
\mathrm{d}s}{s-k}\right]
\]
Due to conditions 1, 4, 5 of Hypothesis \ref{hyp1.1} the function $Ty_{-}$ in
(\ref{eq6.15}) is meromorphic in $\operatorname*{Im}k>0$ with finitely many
simple poles at $\mathrm{i}\kappa_{n}$, and $T\left(  k\right)  \rightarrow
1,k\rightarrow\infty$ , in $\mathbb{C}^{+}$. Compute its residues. It follows
from (\ref{R basic scatt identity}) that we also have
\begin{equation}
T\left(  k\right)  =\frac{2\mathrm{i}k}{W(\psi_{-},\psi_{+})},\label{T}%
\end{equation}
where $W\left(  f,g\right)  :=fg^{\prime}-f^{\prime}g$ is the Wronskians.
Recall that if $k_{0}$ is a zero of $W(\psi_{-},\psi_{+})$ then $\psi
_{+}(x,k_{0})=\mu_{0}\psi_{-}(x,k_{0})$ (linearly dependent) with some
$\mu_{0}\neq0$ , that occurs only for $k_{0}\in\mathrm{i}\mathbb{R}_{+}$ such
that $k_{0}^{2}=-\kappa_{0}^{2}$ , where $-\kappa_{0}^{2}$ is a bound state of
$\mathbb{L}_{q}$ . Next, from the well-known (and easily verifiable) identity
\[
\partial_{k}W(\psi_{-}\left(  x,k\right)  ,\psi_{+}\left(  x,k\right)
)=2k\int_{\mathbb{R}}\psi_{-}\left(  s,k\right)  \psi_{+}\left(  s,k\right)
\mathrm{d}s
\]
one has
\begin{equation}
\left.  \partial_{k}W(\psi_{-}\left(  x,k\right)  ,\psi_{+}\left(  x,k\right)
)\right\vert _{k=\mathrm{i}\kappa_{0}}=2\mathrm{i}\kappa_{0}\mu_{0}^{-1}%
\int_{\mathbb{R}}\psi_{+}(s,\mathrm{i}\kappa_{0})^{2}\mathrm{d}s\text{,}%
\label{Wron rel}%
\end{equation}
which means that $\mathrm{i}\kappa_{0}$ is a simple zero of $W(\psi_{-}%
,\psi_{+})$. It follows from (\ref{Wron rel}) and (\ref{T}) that%
\begin{align*}
\operatorname*{Res}_{k=\mathrm{i}\kappa_{0}}T &  =%
\genfrac{.}{\vert}{}{}{2\mathrm{i}k}{\partial_{k}W(\psi_{-},\psi_{+})}%
_{k=\mathrm{i}\kappa_{0}}=\mathrm{i}\mu_{0}\left(  \int_{\mathbb{R}}\psi
_{+}(s,\mathrm{i}\kappa_{0})^{2}\mathrm{d}s\right)  ^{-1}\\
&  =\mathrm{i}\mu_{0}\Vert\psi_{+}(\cdot,\mathrm{i}\kappa_{0})\Vert_{2}%
^{-2}=\mathrm{i}\mu_{0}c_{0}\text{,}%
\end{align*}
where $c_{0}$ is the (right) norming constant of the bound state $-\kappa
_{0}^{2}$. Therefore,%
\begin{align*}
\operatorname*{Res}_{k=\mathrm{i}\kappa_{n}}T\left(  k\right)  y_{-}\left(
x,k\right)   &  =y_{-}(x,\mathrm{i}\kappa_{n})\operatorname*{Res}%
_{k=\mathrm{i}\kappa_{n}}T\\
&  =\mathrm{i}\mu_{n}y_{-}(x,\mathrm{i}\kappa_{n})c_{n}=\mathrm{i}%
c_{n}\mathrm{e}^{-2\kappa_{n}x}y(x,\mathrm{i}\kappa_{n})\text{,}%
\end{align*}
and taking condition 1 of Hypothesis \ref{hyp1.1} into account one has that
for each fixed $x$
\begin{equation}
T\left(  k\right)  y_{-}\left(  x,k\right)  -1-\sum_{n}\frac{\mathrm{i}%
c_{x},_{n}}{k-\mathrm{i}\kappa_{n}}y(x,\mathrm{i}\kappa_{n})\in H^{2}%
\text{,}\label{LHS}%
\end{equation}
where $c_{x,n}:=c_{n}\mathrm{e}^{-2\kappa_{n}x}$ . Rewrite now (\ref{eq6.15})
in the form%

\begin{align}
&  T\left(  k\right)  y_{-}\left(  x,k\right)  -1-{\sum_{n}}\frac
{\mathrm{i}c_{x,n}}{k-\mathrm{i}\kappa_{n}}y(x,\mathrm{i}\kappa_{n}%
)\nonumber\\
&  =\overline{\left(  y\left(  x,k\right)  -1\right)  }+R_{x}\left(  k\right)
\left(  y\left(  x,k\right)  -1\right)  \nonumber\\
&  +R_{x}\left(  k\right)  -{\sum_{n}}\frac{\mathrm{i}c_{x,n}}{k-\mathrm{i}%
\kappa_{n}}y(x,\mathrm{i}\kappa_{n}).\label{eq6.16}%
\end{align}
Due to (\ref{LHS}), the left hand side of (\ref{eq6.16}). Noticing that the
last term of the right-hand side of (\ref{eq6.16}) is in $H^{2}\left(
\mathbb{C}^{-}\right)  $, the application of the Riesz projection
$\mathbb{P}_{-}$ to (\ref{eq6.16}) yields%
\begin{equation}
\mathbb{P}_{-}(\overline{Y}+R_{x}Y)+\mathbb{P}_{-}R_{x}-\sum_{n}%
\mathrm{i}c_{x,n}\frac{Y(x,\mathrm{i}\kappa_{n})}{k-\mathrm{i}\kappa_{n}}%
-\sum_{n}\frac{\mathrm{i}c_{x,n}}{k-\mathrm{i}\kappa_{n}}=0,\label{eq6.17}%
\end{equation}
where $Y\left(  x,k\right)  :=y\left(  x,k\right)  -1$. Thus the left Jost
solution $\psi_{-}$ is gone from the picture as expected. By condition 1 of
Hypothesis \ref{hyp1.1}, $Y\in H^{2}$ for any $x\in\mathbb{R}$. Since
$\overline{Y}=\mathbb{J}Y$ and by (\ref{eq4.9}) we have
\begin{equation}
\mathbb{P}_{-}\overline{Y}=\mathbb{P}_{-}\mathbb{J}Y=\mathbb{J}\mathbb{P}%
_{+}Y=\mathbb{J}Y.\label{eq6.18}%
\end{equation}
Observing that for any $f\in H^{2}$%
\begin{equation}
\mathbb{P}_{-}\frac{f(\cdot)}{\cdot-\mathrm{i}\kappa}=\mathbb{P}_{-}%
\frac{f(\cdot)-f\left(  \mathrm{i}\kappa\right)  }{\cdot-\mathrm{i}\kappa
}+\mathbb{P}_{-}\frac{f\left(  \mathrm{i}\kappa\right)  }{\cdot-\mathrm{i}%
\kappa}=\mathbb{P}_{-}\frac{f\left(  \mathrm{i}\kappa\right)  }{\cdot
-\mathrm{i}\kappa},\label{P_}%
\end{equation}
we have by (\ref{P_}) that%
\begin{equation}
\sum_{n}\mathrm{i}c_{x,n}\;\frac{Y(x,\mathrm{i}\kappa_{n})}{\cdot
-\mathrm{i}\kappa_{n}}=\mathbb{P}_{-}\sum_{n}\mathrm{i}c_{x,n}\;\frac
{Y(x,\cdot)}{\cdot-\mathrm{i}\kappa_{n}}.\label{eq6.19}%
\end{equation}
Inserting (\ref{eq6.18}) and (\ref{eq6.19}) into (\ref{eq6.17}), we obtain
\[
\mathbb{J}Y+\mathbb{P}_{-}\left(  R_{x}-\sum_{n}\frac{\mathrm{i}c_{x,n}}%
{\cdot-\mathrm{i}\kappa_{n}}\right)  Y=-\mathbb{P}_{-}\left(  R_{x}-\sum
_{n}\frac{\mathrm{i}c_{x,n}}{\cdot-\mathrm{i}\kappa_{n}}\right)  \text{.}%
\]
Applying $\mathbb{J}$ to both sides of this equation yields%
\begin{equation}
(\mathbb{I}+\mathbb{H}(\varphi))Y=-\mathbb{H}(\varphi)1,\label{eq6.20}%
\end{equation}
where $\mathbb{H}(\varphi)$ is the Hankel operator defined in Definition
\ref{def4.1} with symbol%
\begin{equation}
\varphi\left(  k\right)  =\varphi_{x}(k)=R_{x}(k)-\sum_{n}\frac{\mathrm{i}%
c_{x,n}}{k-\mathrm{i}\kappa_{n}}.\label{fi}%
\end{equation}
Due to (\ref{eq6.14}), $\mathbb{J}\varphi=\overline{\varphi}$ and hence
$\mathbb{H}(\varphi)$ is selfadjoint. Note that $\mathbb{H}(\varphi)1$ on the
right hand side of (\ref{eq6.20}) should be interpreted as%
\[
\mathbb{H}(\varphi)1=\mathbb{P}_{+}\bar{\varphi}\in H^{2}\text{,}%
\]
due to condition 4 of Hypothesis \ref{hyp1.1}.

We show that $\mathbb{I}+\mathbb{H}(\varphi)$ is positive definite and hence
(\ref{eq6.20}) is uniquely solvable for $Y(x,k)$ for any real $x$.

We show first that $\mathbb{H}(\varphi)$ with%
\[
\varphi\left(  k\right)  =\frac{-\mathrm{i}c}{k-\mathrm{i}\kappa}%
\]
is semi-positive definite for any positive $c,\kappa$. To this end, for $f\in
H^{2}$ consider the quadratic form $\langle\mathbb{H}(\varphi)f,f\rangle$. By
(\ref{P_}) we have%
\begin{align*}
\langle\mathbb{H}(\varphi)f,f\rangle &  =\langle\mathbb{J}\mathbb{P}%
_{-}\varphi f,f\rangle\\
&  =f\left(  \mathrm{i}\kappa\right)  \langle\mathbb{J}\frac{-\mathrm{i}%
c}{\cdot-\mathrm{i}\kappa},f\rangle=\mathrm{i}cf\left(  \mathrm{i}%
\kappa\right)  \langle\frac{1}{\cdot+\mathrm{i}\kappa},f\rangle\\
&  =\mathrm{i}cf\left(  \mathrm{i}\kappa\right)  \overline{\langle f,\frac
{1}{\cdot+\mathrm{i}\kappa}\rangle}=2\pi c\left\vert f\left(  \mathrm{i}%
\kappa\right)  \right\vert ^{2}\geq0.
\end{align*}
Here at the last step we used the Cauchy formula. We can now conclude that the
operator%
\[
\mathbb{H}\left(  \sum_{n}\frac{-\mathrm{i}c_{x,n}}{\cdot-\mathrm{i}\kappa
_{n}}\right)  =\sum_{n}\mathbb{H}\left(  \frac{-\mathrm{i}c_{x,n}}%
{\cdot-\mathrm{i}\kappa_{n}}\right)
\]
is semi-positive definite for any $x$.

Thus, the problem now boils down to showing that $\mathbb{I}+\mathbb{H}\left(
R_{x}\right)  $ is positive definite for any $x$. It follows from conditions 2
and 3 of Hypothesis \ref{hyp1.1} that%
\begin{align*}
\sup_{j}\frac{1}{2}|R_{x}(\omega_{j}+0)-R_{x}(\omega_{j}-0)| &  =\sup_{j}%
\frac{1}{2}|\mathrm{e}^{2\mathrm{i}\omega_{j}x}(R(\omega_{j}+0)-R(\omega
_{j}-0))|\\
&  =\sup_{j}\frac{1}{2}|R(\omega_{j}+0)-R(\omega_{j}-0)|<1
\end{align*}
and by Theorem \ref{Power thm}
\[
\sigma_{ess}\left(  \mathbb{H}(R_{x})\right)  =\cup_{j}\left[  -\alpha
_{j},\alpha_{j}\right]  \text{,}%
\]
where
\[
\alpha_{j}=\frac{1}{2}\left\vert R\left(  \omega_{j}+0\right)  -R\left(
\omega_{j}-0\right)  \right\vert \text{.}%
\]
One now concludes that $\lambda=\pm1\notin\sigma_{ess}\left(  \mathbb{H}%
(R_{x})\right)  $ and therefore $\lambda=-1$ could only be an eigenvalue of
finite multiplicity. It remains to show that it is not the case. We proceed by
contradiction \ Assume $\lambda=-1$ is an eigenvalue of finite multiplicity
and let $f\in H^{2}$ be the associated normalized eigenfunction. We then have
by (\ref{eq4.9})
\[
\langle\mathbb{H}(R_{x})f,f\rangle=\langle R_{x}f,\mathbb{P}_{-}%
\mathbb{J}f\rangle=\langle R_{x}f,\mathbb{J}f\rangle
\]
and hence by the Cauchy inequality
\begin{align}
\left\vert \langle\mathbb{H}(R_{x})f,f\rangle\right\vert ^{2} &  \leq
\int_{\mathbb{R}}\left\vert R\left(  k\right)  \right\vert ^{2}\left\vert
f\left(  k\right)  \right\vert ^{2}\mathrm{d}k\;\left\Vert f\right\Vert
_{2}^{2}\label{eig}\\
&  =\int_{\mathbb{R}}\left\vert R\left(  k\right)  \right\vert ^{2}\left\vert
f\left(  k\right)  \right\vert ^{2}\mathrm{d}k\nonumber\\
&  =\int_{S}\left\vert R\left(  k\right)  \right\vert ^{2}\left\vert f\left(
k\right)  \right\vert ^{2}\mathrm{d}k+\int_{\mathbb{R}\setminus S}\left\vert
R\left(  k\right)  \right\vert ^{2}\left\vert f\left(  k\right)  \right\vert
^{2}\mathrm{d}k<\Vert f\Vert_{2}^{2}=1,\nonumber
\end{align}
where $S$ is a set of positive Lebesgue measure where $\left\vert R\right\vert
<1$ a.e. Here we have used the fact that $f\in H^{2}$ and hence cannot vanish
on $S$ . The inequality (\ref{eig}) implies that $\left\vert \lambda
\right\vert <1$ which is a contradiction.

\section{Applications to WvN type potentials \label{apps}}

Theorem \ref{MainThm} applies to a variety of oscillatory potentials decaying
as $O(1/x)$ but without understanding the zero energy behavior of scattering
data results could only be partial. To describe them we review some results on
WvN type potentials following \cite{Hintonetal1991, Klaus91} and adjust them
to our setting.

Consider a continuous potential of the form
\begin{equation}
q\left(  x\right)  =q_{\gamma}\left(  x\right)  +O(x^{-2}),x\rightarrow
\pm\infty,\label{WvN}%
\end{equation}
where $q_{\gamma}\left(  x\right)  \text{{}}$ is given by (\ref{pure WvN}).
Clearly $q$ is square integrable (but not even integrable) and hence the
classical short-range techniques do not apply. Nevertheless this potential
supports two Jost solutions\footnote{It is been known since the 1950s if not
earlier.} $\psi_{\pm}\left(  x,k\right)  \text{{}}$ analytic for
$\operatorname*{Im}k>0$ and continuous up to the real line except for
$k=\pm\omega$ (and possibly $0$) where they blow up to the order of $\gamma$:
\begin{equation}
\psi_{\pm}\left(  x,k\right)  \sim\frac{u_{\pm}\left(  x\right)  }{\left(
k-\omega\right)  ^{\gamma}},k\rightarrow\omega,\operatorname*{Im}%
k\geq0.\label{asympt}%
\end{equation}
Due to the symmetry the same behavior takes place of course at $-\omega$. In
(\ref{asympt}) the branch cut is taken along $\mathbb{R}_{-}$ and $u_{\pm
}\left(  x\right)  $ are solutions to $\mathbb{L}_{q}u=\omega^{2}u$ determined
by the asymptotics
\begin{equation}
u_{\pm}\left(  x\right)  \sim c_{\pm}x^{-\gamma}\cos\omega x,\text{\ \ }%
x\rightarrow\pm\infty,\label{u+_}%
\end{equation}
with some (complex) constants $c_{\pm}$. Thus $\psi_{\pm}\left(  x,k\right)  $
blow up to the order $\gamma$ at $k=\pm\omega$ (and possibly also at $k=0$)
but are continuous elsewhere. For large $|k|$ the behavior is the same as that
of the short-range case $\psi_{\pm}\left(  x,k\right)  =\mathrm{e}%
^{\pm\mathrm{i}k}(1+O(1/k)),\operatorname*{Im}k\geq0$. Note that in the
short-range case $\psi_{\pm}\left(  x,k\right)  $ are continuous on the entire
real line. The points $\pm\omega$ are therefore can be called spectral
singularities. In the literature however the energy $\omega^{2}$ is referred
to as a WvN resonance. Generically, $u_{\pm}(x)$ are linearly independent but
for specially chosen $q$'s they may become linearly dependent meaning that the
Schrodinger equation $\mathbb{L}_{q}u=\omega^{2}u$ has then a solution which
is square integrable if $\gamma>1/2$ and decaying (but not $L^{2}$) if
$\gamma\leq1/2$ (this will be our case). In the former $\omega^{2}$ is a bound
state of $\mathbb{L}_{q}$ embedded into the a.c. spectrum and in the latter
$\omega^{2}$ is called a half-bound state. Note that both bound/half-bound
states are very unstable (while a resonance is) and can be destroyed by an
arbitrarily small perturbation\footnote{The situation is reminiscent of
generic vs exceptional alternative in the short-range case when the transition
$T(k)$ is generically vanishes linearly a $k\rightarrow0$ but exceptional
$q$'s support nontrivial transition at zero momentum. This situation is
extremely unstable and there is no way to tell generic and exceptional
potentias apart. On the bright side, this phenomenon almost never matters.}.
Thus we assume that $W\left(  u_{-},u_{+}\right)  \neq0$.

We now consider the transmission coefficients $T$ and $R$. It immediately
follows from (\ref{asympt}) that%
\[
W\left(  \psi_{-},\psi_{+}\right)  \sim\frac{W\left(  u_{-},u_{+}\right)
}{\left(  k-\omega\right)  ^{2\gamma}},\;k\rightarrow\omega,\operatorname*{Im}%
k\geq0\text{,}%
\]
and by (\ref{T}) we have
\begin{equation}
T\left(  k\right)  =\frac{2\mathrm{i}\omega}{W\left(  u_{-},u_{+}\right)
}\left(  k-\omega\right)  ^{2\gamma}=O\left(  k-\omega\right)  ^{2\gamma
},\;k\rightarrow\omega,\operatorname*{Im}k\geq0.\label{asy}%
\end{equation}
Thus $T(k)$ vanishes to the order of $2\gamma$ as $k\rightarrow\pm\omega$.
Recall that in the short-range case $T(k)$ vanishes only at zero. Turn now to
$R(k)$. For some real neighborhood of $\omega$ (we again use
\cite{Hintonetal1991, Klaus91})
\begin{equation}
\overline{\psi_{+}\left(  x,k\right)  }/\psi_{+}\left(  x,k\right)
\sim\operatorname*{Sgn}A\;\mathrm{e}^{-\mathrm{i}\pi\gamma\operatorname*{Sgn}%
\left(  k-\omega\right)  },k\rightarrow\omega.\label{jump}%
\end{equation}
Dividing (\ref{R basic scatt identity}) by $\psi_{+}$ and taking into account
(\ref{jump}) yields
\begin{align*}
T\left(  k\right)  \psi_{-}\left(  x,k\right)  /\psi_{+}\left(  x,k\right)
&  =\overline{\psi_{+}\left(  x,k\right)  }/\psi_{+}\left(  x,k\right)
+R\left(  k\right)  \\
&  \sim\operatorname*{Sgn}A\;\mathrm{e}^{-\mathrm{i}\pi\gamma
\operatorname*{Sgn}\left(  k-\omega\right)  }+R\left(  k\right)  \text{.}%
\end{align*}
On the other hand, by (\ref{asympt}) and (\ref{asy}) one concludes that
\[
R\left(  k\right)  +\operatorname*{Sgn}A\;\mathrm{e}^{-\mathrm{i}\pi
\gamma\operatorname*{Sgn}\left(  k-\omega\right)  }=O\left(  k-\omega\right)
^{2\gamma},\text{\ \ \ }k\rightarrow\omega\text{,}%
\]
and hence
\begin{equation}
\lim_{k\rightarrow\omega}\mathrm{e}^{\mathrm{i}\pi\gamma\operatorname*{Sgn}%
\left(  k-\omega\right)  }R\left(  k\right)  =-\operatorname*{Sgn}%
A.\label{sign}%
\end{equation}
Note that the subtle fact that (\ref{asympt}) holds as $k\rightarrow\omega$
also along the real line (not only tangentially from $\mathbb{C}^{+}$) was
crucially used as (\ref{R basic scatt identity}) holds in general on the real
line only.

Let us discuss now the decay property of $R(k)$. To this end, we recall the
well-known KdV second conservation law (aka the second Zakharov-Faddeev trace
formula or the Shabat-Zakharov sum rule)
\begin{equation}
\frac{2\pi}{3}\sum_{n}\kappa_{n}^{3}+\int_{\mathbb{R}}k^{2}\log(1-|R\left(
k\right)  |^{2})^{-1}\mathrm{d}k=\frac{\pi}{8}\int_{\mathbb{R}}q\left(
x\right)  ^{2}\mathrm{d}x.\label{trace formula}%
\end{equation}
Since $q\in L^{2}$, (\ref{trace formula}) immediately implies that $R(k)$ must
be at least square integrable (indeed $\log\left(  1-\left\vert R\right\vert
^{2}\right)  ^{-1}\geq\left\vert R\right\vert ^{2}$), which is the same rate
of decay as in the short-range case. Incidentally, due to (\ref{eq6.14}) this
means that $T(k)=1+O(1/k)$ as $k\rightarrow\infty$ which is in agreement with
the short-range scattering.

Thus, the reflection coefficient $R$ is continuous away from $\pm\omega$ (and
possibly $0$), has at $\pm\omega$ jump discontinuity of size $2|\sin\pi
\gamma|$, and decays at $o(1/k)$. Recalling the properties of $R$ in the
short-range case we see that the appearance of discontinuities seems to be the
only difference. However this phenomenon alone makes the machinery of the
inverse scattering break down in a very serious way. As was already mentioned
in the introduction, it remains a good open problem.

Yet another circumstance is a poor understanding of the zero energy behavior
(also stated in \cite{Klaus91} as an open problem). It is shown in
\cite{Klaus82} that $q_{\gamma}$ (a pure NvW type potential) has finite
negative spectrum if $\gamma<\sqrt{1/2}$ but if $\gamma\geq\sqrt{1/2}$ the
negative spectrum (necessarily discrete) is infinite, accumulating to zero.
Recall that in the short-range inverse scattering every (negative) bound state
requires a norming constant. Since those norming constants $c_{n}$ in the KdV
context determine locations of solitons (which can be arbitrary) the hope that
$c_{n}$ may have a pattern of behavior as $n\rightarrow\infty$ is not
justified in general\footnote{Besides, the results of \cite{Yafaev81} suggest
that the behavior of the Jost solution at zero could be quite messy.}. We
believe that there is no hope in trying to adapt the classical inverse
scattering to the setting of infinite negative spectrum in general and instead
a totally different approach is required, which would also handle the rough
behavior of the Jost function at zero (work in progress). In this contribution
we are focus on the effect of positive resonances on the inverse scattering,
which is already very important, and address zero resonance and infinite
negative spectrum elsewhere.

Thus we assume that our $q$ is such that $k =0$ is a regular point of the
spectrum of $\mathbb{L}_{q}$. In other words, $\psi_{ \pm} \left( x ,k\right)
$ are continuous\footnote{In fact, boudedness of $\left\vert k\right\vert
^{1/2 -\varepsilon} \psi_{ \pm} \left( x ,k\right) $ at $k =0$ for some
$\varepsilon>0\text{ would be sufficient.}$} as $k \rightarrow\omega
\operatorname*{\text{in }Im} k \geq0\text{}$. Since the (necessarily
imaginary) zeros of the Jost solution interlace with $\mathrm{i} \kappa_{n}$,
this assumption rules out negative infinite spectrum.

Thus one can now see that if $\gamma<1/2$ then all conditions of Hypothesis
\ref{hyp1.1} are satisfied and we arrive at

\begin{theorem}
\label{Thm on WvN} Let a continuous $q(x)$ be of the form (\ref{WvN}) with
some real $A$ and positive $\omega$ such that $\gamma=\left\vert
A/(4\omega)\right\vert \in\left(  0,1/2\right)  $. Suppose that (a)
$\omega^{2}$ is not a half bound state of $\mathbb{L}_{q}$ and (b) $0$ is a
regular point of $\sigma(\mathbb{L}_{q})$. Then the set $S_{q}=\left\{
R,(\kappa_{n},c_{n})\right\}  $ determines $q$ uniquely from the relation
\begin{equation}
Y\left(  x,\cdot\right)  =-(\mathbb{I}+\mathbb{H}(\varphi_{x}))^{-1}%
\mathbb{H}(\varphi_{x})1.\label{Y}%
\end{equation}
Furthermore, $\omega>0$ is found from the equation $\left\vert R\left(
\omega\right)  \right\vert =1$; $\gamma$ and the sign of $A$ is determined
from (\ref{sign}).
\end{theorem}

We did not specify how one can obtain $q$ from (\ref{Y}). The obvious one is
to find $q$ directly from (\ref{SE}). The most common one appears to be by
\begin{equation}
q \left( x\right)  = - \partial_{x}\lim2 \mathrm{i} k (\mathbb{I} +\mathbb{H}
(\varphi_{x}))^{ -1} \mathbb{H} (\varphi_{x}) 1 ,k \rightarrow\infty
,\operatorname*{Im} k \geq0.\label{q}%
\end{equation}
The most general one is
\begin{equation}
q \left( x\right)  =\lim2 k^{2} \left( G \left( x ,k^{2}\right)  -1\right)
,k^{2} \rightarrow\infty,\operatorname*{Im}  k^{2} \geq0 ,\label{G}%
\end{equation}
where
\[
G \left( x ,k^{2}\right)  =\frac{\psi_{ -} \left( x ,k\right)  \psi_{ +}
\left( x ,k\right) }{W \left( \psi_{ +} \left( x ,k\right)  ,\psi_{ -} \left(
x ,k\right) \right) }%
\]
is the diagonal Green's function. Note that (\ref{G}) is not sensitive to
conditions on $q$ once we replace $\psi_{ \pm}$ with Weyl solutions.

\begin{corollary}
Under the conditions of Theorem \ref{Thm on WvN}, the problem (\ref{KdV}%
)-(\ref{KdVID}) has a unique solution given by (\ref{q}) where $\varphi_{x}$
is replaced with
\begin{equation}
\varphi_{x ,t} (k) =R (k) \mathrm{e}^{8 \mathrm{i} k^{3} t +2 \mathrm{i} k x}
-\sum_{n}\frac{\mathrm{i} c_{n}}{k -\mathrm{i} \kappa_{n}} \mathrm{e}^{8
\kappa_{n}^{3} t -2 \kappa_{n} x} .\label{eq9.9}%
\end{equation}

\end{corollary}

\begin{proof}
By the Zakharov-Shabat dressing method \cite{NovikovetalBook}, multiplying
scattering data by $\mathrm{e}^{8 \mathrm{i} k^{3} t}$ implies the time
evolution (the KdV flow) $q \left( x\right)  \rightarrow q \left( x ,t\right)
$ . Thus $q \left( x ,t\right) $ given by (\ref{q}) with $\varphi_{x}$
replaced with $\varphi_{x ,t}$ solves (\ref{KdV}) with initial data $q \in
L^{2}$ . Therefore, by the Bourgain theorem \cite{Bourgain93} $q \left( x
,t\right) $ is the solution to (\ref{KdV})-(\ref{KdVID}).
\end{proof}

The set of potentials in Theorem \ref{Thm on WvN} is not empty but its
description in terms of potentials is likely impossible. Constructing specific
potentials can be done as follows. We take the spectral measure of the free
($q=0$) half-line Schrodinger operator with a Dirichlet boundary condition at
zero and perturb it in a small neighborhood of $\omega^{2}$ to produce a
singularity of order $\gamma<1/2$ at $\omega^{2}$ of the Weyl m-function. This
procedure gives rise a potential (necessarily $L^{2}$ due to
\cite{KillipSimon2008}) on $\mathbb{R}_{+}$ which we continue to
$\mathbb{R}_{-}$ as an even function. The potential constructed this way will
be subject to the conditions of Theorem \ref{Thm on WvN}. We employ a similar
idea in the recent \cite{RyNON21} to construct a potential,%
\[
q\left(  x\right)  =\left\{
\begin{array}
[c]{cc}%
q_{0}\left(  x\right)  \text{,} & x\geq0\\
q_{0}\left(  -x\right)  \text{,} & x<0
\end{array}
\right.  \text{,}%
\]
where%
\[
q_{0}\left(  x\right)  =-2\partial_{x}^{2}\log\left(  1+\rho x-\left(
\rho/2\right)  \sin2x\right)  ,\text{\ \ \ }x\geq0,\rho>0\text{,}%
\]
that has two Jost solutions given by%
\[
\psi_{\pm}\left(  x,k\right)  =\left\{  1\pm\left(  \frac{\mathrm{e}%
^{\pm\mathrm{i}x}}{k+1}-\frac{\mathrm{e}^{\mp\mathrm{i}x}}{k-1}\right)
\frac{\rho\sin x}{1+\rho\left\vert x\right\vert -\left(  \rho/2\right)
\sin2\left\vert x\right\vert }\right\}  \mathrm{e}^{\pm\mathrm{i}kx},\pm
x\geq0.
\]
Apparently, this potential has the asymptotic behavior (\ref{WvN}) but
$\psi_{\pm}\left(  x,k\right)  $ are clearly continuous at $k=0$ (for each $x$
). The potential $q(x)$ is different at $\pm\infty$ from the WvN potential
$q_{\gamma}$ with $\gamma=2$ by $O(x^{-2})$ and its negative spectrum has only
one bound state. A different construction for $\gamma=1/2$ is given in
\cite{NovikovKhenkin84} based on a subtle limiting procedure in the
Gelfand-Levitan-Marchenko equation. This procedure also yields potentials
having a smooth spectral measure at zero and decaying like $q_{\gamma}$ with
$\gamma=1/2$. In fact, more than one resonance point is allowed but the
techniques do not yield a description of this class in terms of $q$'s either.
Since both constructions produce strong WvN potentials (i.e. $\gamma\geq1/2$)
the reader should be convinced that weak WvNs are not any worse. We however
have now a totally different approach (work in progress) based on analyzing
the effect of altering the spectral measure on a small interval $\left(
-\varepsilon,\varepsilon\right)  $ to remove infinitely many (small) negative
bound states and smoothen its behavior at the edge of the a.c. spectrum. We
conjecture that such an alteration results in a $O(x^{-2})$ perturbation for a
broad class of potentials.

Theorem \ref{Thm on WvN} can be extended to a finite sum of potentials of type
(\ref{WvN}) with a new interesting feature. The set of WvN resonances is
merely $\left\{  \omega_{j}^{2}\right\}  $ and each resonance produces a
jump\footnote{At each $\omega_{j}$ the Jost solution exhibits a singularity of
type (\ref{asympt}) but $\psi$ is still in $H^{2}$ as long as $\gamma_{j}<1/2$
.} of $R$ at $\pm\omega_{j}$ of size $\left\vert 2\sin\pi\gamma_{j}\right\vert
$. This puts us in the setting of the recent deep paper \cite{PushYaf2015},
where the spectral and scattering theory for self-adjoint Hankel operators
$\mathbb{H}\left(  \varphi\right)  $ with piecewise continuous $\varphi$ is
developed. Under a mild extra condition, each jump of $\varphi$ gives rise to
an interval of the a.c. spectrum of $\mathbb{H}\left(  \varphi\right)  $ (c.f.
Theorem \ref{Power thm}) the authors construct \emph{wave operators} realizing
unitary equivalence of $\mathbb{H}\left(  \varphi\right)  $ and the orthogonal
sum of simple model\ Hankel operators responsible for each jump similarly to
the famous Faddeev's solution of the \emph{three particle quantum problem}. In
our case it means that the a.c. part of $\mathbb{H}\left(  \varphi
_{x,t}\right)  $ is unitary equivalent to the sum of orthogonal Hankel
operators and each WvN resonance produces a term in the KdV solution. This
situation is very similarly to the nonlinear superposition of solitons and one
should expect an analog of nonlinear superposition for radiation waves (i.e.
propagating to $-\infty$ ) with phase velocities $12\omega_{j}^{2}$ (similarly
to how each negative bound state $-\kappa_{n}^{2}$ produces a soliton with
velocity $4\kappa_{n}^{2}$ ). It is worth noticing that a similar phenomenon
is studied in \cite{Matveev2002} in the context of singular (i.e. with a local
double pole singularity) WvN type potential commonly referred to as positons.
Note that positons do not interact.

Similarly, Theorem \ref{Thm on WvN} can be extended along the lines to a more
general and physically relevant case of a decaying "periodic" structure
modelled by
\[
q \left( x\right)  =p \left( x\right) /x\text{,}%
\]
where $p \left( x\right) $ is a small enough periodic function with
zero-average. The spectral theoretical basis for this is developed in
\cite{Hintonetal1991,Klaus91} where $p$ is expanded into the Fourier series
producing an infinite sum of WvN potentials with Fourier frequencies
$\omega_{j}$ and $\left( \gamma_{j}\right)  \in l^{1}$ , which guarantees
convergences. The KdV flow promises some fascinating dynamics.

If $\gamma=1/2$ then the essential spectrum of $\mathbb{H}\left(  R\right)  $
fills $\left[  -1.1\right]  $ and $\mathbb{I}+\mathbb{H}\left(  R\right)  $ is
no longer boundedly invertible creating a serious problem to our method. It is
interesting to note that this problem occurs only if $\gamma=1/2+n$ for any
natural $n$ but $\psi$ , which still shows the behavior (\ref{asympt}), is no
longer in $H^{2}$ and our approach needs serious modifications. We believe
that instead of invertibility of $\mathbb{I}+\mathbb{H}\left(  R_{x}\right)  $
we need to study invertibility of $\mathbb{T}\left(  T/\overline{T}\right)
+\mathbb{H}\left(  R_{x}\right)  $ where $\mathbb{T}\left(  \varphi\right)  $
is the Toeplitz operator with symbol $\varphi$ .

In the conclusion of this section we mention that condition (a) of Theorem
\ref{Thm on WvN} is actually unnecessary but the arguments become more
complicated. Note that essentially any compactly supported perturbation turns
a half-bound state into a resonance.

\section{Application to a matrix Riemann-Hilbert problem \label{RHP section}}

In this section we apply Theorem \ref{MainThm} to $2 \times2$ matrix
Riemann-Hilbert problem that arrises in the Riemann-Hilbert problem approach
to the IST

As is well-known, one can rewrite the time-evolved
(\ref{R basic scatt identity}) as a meromorphic vector Riemann-Hilbert problem
(see e.g. \cite{GT09}) with the \emph{jump matrix}
\begin{equation}
V =\left(
\begin{array}
[c]{cc}%
1 -\left\vert R\right\vert ^{2} & -\overline{R}_{x ,t}\\
R_{x ,t} & 1
\end{array}
\right)  ,\label{jump matrix}%
\end{equation}
where as before $R_{x ,t} \left( k\right)  =\exp\left( 8 \mathrm{i} k^{3} t +2
\mathrm{i} k x\right) $ with real $x ,t$ . In the language of the
Riemann-Hilbert problem, we can claim that loosely speaking the IST works
smoothly iff
\begin{equation}
V =V_{ -} V_{ +}\label{canonical factor}%
\end{equation}
with a unique choice of $2 \times2$ matrices $V_{ \pm}$ subject to $V_{ \pm}
-I \in H^{2} \left( \mathbb{C}^{ \pm}\right) $ . Such factorization
(\ref{canonical factor}) is referred to as a canonical $L^{2}$ factorization
\cite{ClanceyGohbergBOOK} and is called a matrix Riemann-Hilbert problem in
the IST community. Note that if $1 -\left\vert R \left( k\right) \right\vert
^{2} >0$ then $\operatorname*{Re} V$ is positive definite and
(\ref{canonical factor}) holds \cite{ClanceyGohbergBOOK}. Such situation
occurs in the modified KdV case \cite{DeiftZhou1993} but not in the KdV case
as generically $\left\vert R \left( 0\right) \right\vert =1$ even in the
short-range case. In our case $1 -\left\vert R \left( \omega_{j}\right)
\right\vert ^{2} =0$ . However, we still have (\ref{canonical factor}). More specifically

\begin{theorem}
\label{factorization thm} Let a reflection coefficient $R$ in (\ref{jump}) be
such $\left\vert R \left( k\right) \right\vert \leq1$ , and $\left\vert R
\left( k\right) \right\vert <1$ on a set $S$ of positive measure. Suppose that
$R$ is piecewise continuous on the closed $\mathbb{R}$ away from some points
$\left\{  \pm\omega_{j}\right\} $ (may be infinitely many). Then the matrix
Riemann-Hilbert problem
\begin{equation}%
\begin{array}
[c]{c}%
V \left( k\right)  =V_{ -} \left( k\right)  V_{ +} \left( k\right)  ,k
\in\mathbb{R}\text{,}\\
V_{ \pm} -I \in H^{2} \left( \mathbb{C}^{ \pm}\right) \text{,}%
\end{array}
\label{matrix RHP}%
\end{equation}
has a unique solution iff
\[
\sup_{j}\frac{1}{2} \left\vert R \left( \omega_{j} +0\right)  -R \left(
\omega_{j} -0\right) \right\vert <1.
\]

\end{theorem}

\begin{proof}
Our arguments are based on an important result of \cite{LitSpit82}, that being
adjusted to our situation, says that the problem (\ref{matrix RHP}) has unique
solution iff $-1\notin\sigma_{ess}\left(  \mathbb{H}\left(  R_{x,t}\right)
\right)  $. But as we have shown in the proof of Theorem \ref{MainThm} it is
always the case under conditions 2 as long as 3 of Hypothesis \ref{hyp1.1} is satisfied.
\end{proof}

In the conclusion of the paper we would like to emphasize that the
Riemann-Hilbert problem statement of the IST is extremely powerful tool in the
asymptotic analysis of solutions to integrable systems (see
\cite{DeiftZhou1993} in the mKdV case and \cite{GT09} in the KdV case).
However, it does not have that edge in the circle of problems we are concerned
with and our Hankel operator approach works well instead.

\section{Acknowledgment}

The first author acknowledges the support provided by CONACYT, Mexico via the
project Ciencia de Frontera FORDECYT-PRONACES/61517/2020 and the work was
performed in part at the Regional mathematica center of the Southern Federal
University with the support of the Ministry of Science and Higher Education of
Russia, agreement 075-02-2021-1386. 

The second author acknowledges partial support from NSF under grant DMS 1716975.

The authors are grateful to Alexander Pushnitsky for drawing our attention to
\cite{PushYaf2015} which resulted in this work. We are also thankful to Ilya
Spitkovsky for numerous discussions leading to Section \ref{RHP section}. And
last but the least, we are grateful to the referees for numerous comments and
questions leading to a substantial improvement of the paper.

\end{document}